\documentclass[11pt]{article}
\usepackage{amssymb}
\usepackage{amsthm}
\usepackage{latexsym}
\usepackage{amsmath}
\usepackage{mathabx}
\usepackage[all]{xy}
\usepackage{titling}
\thanksmarkseries{arabic}
\newtheorem{theo}{Theorem}[section]
\newtheorem{prop}[theo]{Proposition}
\newtheorem{cor}[theo]{Corollary}
\newtheorem{lem}[theo]{Lemma}
\theoremstyle{definition}
\newtheorem{defi}[theo]{Definition}
\newtheorem{defiprop}[theo]{Definition/Proposition}
\newtheorem{exa}[theo]{Example}
\newtheorem{rem}[theo]{Remark}
\newtheorem{question}[theo]{Question}
\newtheorem{conj}[theo]{Conjecture}
\numberwithin{equation}{section}

\newcommand{\F}{{\mathbb F}}

\newcommand{\Z}{{\mathbb Z}}

\newcommand{\C}{{\mathbb C}}

\newcommand{\cE}{{\mathcal E}}

\newcommand{\cG}{{\mathcal G}}

\newcommand{\cP}{{\mathcal P}}

\newcommand{\cQ}{{\mathcal Q}}
\newcommand{\cB}{{\mathcal B}}

\newcommand{\cS}{{\mathcal S}}

\newcommand{\cU}{{\mathcal U}}

\newcommand{\mm}{{\mathfrak m}}
\newcommand{\veps}{{\varepsilon}}
\newcommand{\Char}{\mbox{${\rm char}$}}

\newcommand{\ann}{\mbox{\rm ann}}
\newcommand{\soc}{\mbox{\rm soc}}
\newcommand{\rad}{\mbox{\rm rad}}
\newcommand{\im}{\mbox{\rm im}\,}

\newcommand{\Rhat}{\mbox{$\widehat{R}$}}

\newcommand{\lcm}{\mbox{\rm lcm}}

\newcommand{\Tr}{\mbox{\rm Tr}}

\newcommand{\wt}{{\rm wt}}
\newcommand{\wts}{\mbox{${\rm wt}_{\rm s}$}}
\newcommand{\wtH}{\mbox{${\rm wt}_{\rm H}$}}
\newcommand{\ds}{\mbox{${\rm d}_{\rm s}$}}

\newcommand{\T}{\mbox{$\!^{\sf T}$}}
\newcommand{\group}[1]{\mbox{$\langle{#1}\rangle$}}
\newcommand{\inner}[2]{\mbox{$\langle{#1}\,|\,{#2}\rangle$}}
\newcommand{\inners}[2]{\mbox{$\langle{\,{#1}\,}|\,{#2}\,\rangle_{\rm s}$}}

\newcommand{\ov}[1]{\mbox{$\overline{#1}$}}
\newcommand{\GL}{\mathrm{GL}}
\newcommand{\SL}{\mathrm{SL}}
\newcommand{\Symp}{\mathrm{Symp}}
\newcommand{\MonSL}{\mbox{$\mathrm{Mon}_{\mathrm{SL}}$}}
\newcommand{\diag}{\textup{diag}}

\newcommand{\dist}{\textup{dist}}
\newcommand{\cBn}{\mbox{$\cB^{\otimes n}$}}

\newcommand{\Smallfourmat}[4]{\mbox{$\left(\begin{smallmatrix}{#1}&{#2}\\{#3}&{#4}\end{smallmatrix}\right)$}}

\newcounter{alp}
\newcounter{ara}
\newcounter{rom}

\newenvironment{alphalist}{\begin{list}{(\alph{alp})\hfill}{\usecounter{alp}
     \topsep-.2ex \labelwidth.6cm \leftmargin.6cm \labelsep0cm
     \rightmargin0cm \parsep0ex \itemsep0ex}}{\end{list}}
\newenvironment{arabiclist}{\begin{list}{(\arabic{ara})\hfill}{\usecounter{ara}
     \topsep0.4ex \labelwidth.6cm \leftmargin.6cm \labelsep0cm
     \rightmargin0cm \parsep0ex \itemsep0ex}}{\end{list}}

\begin{document}
\title{On Quantum Stabilizer Codes derived from\\ Local Frobenius Rings}
\date{\today}
\author{Heide Gluesing-Luerssen\thanks{HGL was partially supported by the grant \#422479 from the Simons Foundation.
  HGL and TP are with the Department of Mathematics, University of Kentucky, Lexington KY 40506-0027, USA;
\{heide.gl, tefjol.pllaha\}@uky.edu.}\ \ and Tefjol Pllaha$^1$}

\maketitle

{\bf Abstract:}
 In this paper we consider stabilizer codes over local Frobenius rings.
 First, we study the relative minimum distances of a stabilizer code and its reduction onto the residue field.
 We show that for various scenarios, a free stabilizer code over the ring does not underperform the according stabilizer code over the  field.
This leads us to conjecture that the same is true for all free stabilizer codes.
 Secondly, we focus on the isometries of stabilizer codes.
 We present some preliminary results and introduce some interesting open problems.


\section{Introduction}
Quantum error-correcting codes have gained a lot of attention during the last two decades.
This is particularly true for quantum stabilizer codes, which are sometimes regarded as the quantum equivalent to linear codes in classical coding theory.
While the first examples of quantum error-correcting codes were in fact stabilizer codes,
the general construction was introduced by Gottesman~\cite{Go96} and Calderbank et al.~\cite{CRSS97}.

Quantum stabilizer codes are particularly appealing because of their close relation to classical codes which allows to analyze them with classical methods.
More precisely, as observed by Calderbank et al.~\cite{CRSS98}, the quantum  error-correcting performance of a quantum stabilizer code is determined by the properties
of a self-orthogonal classical binary code (called stabilizer code), where self-orthogonality refers to a certain symplectic inner product.
This observation has led to the construction of many good codes.
Subsequently, the entire setting has been generalized by Ashikhmin/Knill~\cite{AsKn01} to quantum stabilizer codes derived from non-binary fields.
Ketkar et al.\ in~\cite{KKKS06} derive, among other things, various classes of
codes such as quantum Hamming codes, quantum quadratic residue codes etc.

The idea of generalizing the setting of quantum stabilizer codes can be driven even further by starting with classical codes over finite rings.
This has been initiated by Nadella/Klappenecker in~\cite{NaKl12}, where they mainly focus on finite chain rings.
Their main result shows that a quantum stabilizer code derived from a free code over a finite chain ring does not have a larger distance than the
associated code over the residue field (as always for stabilizer codes, the distance is a relative distance based on the symplectic weight; see Definition~\ref{D-Dist}).

In this paper we consider the same question for codes over finite local commutative Frobenius rings.
More precisely, we investigate whether such a code can have a larger distance than the code obtained by reduction to the vector space over the residue field.
As our work will show, local Frobenius rings appear to be the natural setting for this question because -- just like chain rings -- they have a principal socle, which
then is the annihilator of the maximal ideal and that is all that is needed in order to derive a close relationship between the code and its reduction.
However, unfortunately, we can provide only partial answers to our question, but we conjecture that stabilizer codes over local Frobenius rings do not perform
worse than their reduction over the residue field.
In addition, for free codes we can settle various cases where the distance equals the one of its reduction.

In order to carry out these investigations we first have to develop the proper setup. This will be done in the next section, where we introduce the quantum space and error
basis associated to a given local Frobenius ring.
The definition of the associated Pauli group needs particular care so that stabilizer codes are exactly the self-orthogonal codes with respect to the symplectic inner
product stemming from the non-commutativity of the error basis.
More precisely, this requires the phases of the Pauli operators be drawn from a specific group of primitive roots of unity.
All of this is carried out in Section~\ref{SS-StabCode}.
In Section~\ref{SS-RedDim} we study stabilizer codes in~$R^{2n}$, where~$R$ is a local Frobenius ring, and their reductions modulo the maximal ideal~$\mm$ of~$R$.
The latter are stabilizer codes over the field $R/\mm$.
We show that the dimension of free codes is invariant under reduction and derive a standard form for free codes and their duals.
The latter makes use of symplectic isometries, i.e., linear maps that preserve the symplectic inner product and the symplectic weight -- and thus self-orthogonality.
In Sections~\ref{SS-Dist} and~\ref{SS-Exa} we compare the relative distance of a code with that of its reduction and provide various examples.
Finally, in the last section we return to symplectic isometries.
We describe explicitly the symplectic isometries on the entire space $R^{2n}$.
This allows us to show that symplectic isometries between stabilizer codes do not necessarily extend to such isometries on the entire space (even in the binary case).
This result indicates that a classification of stabilizer codes into symplectic isometry classes is by no means obvious and therefore leads to open problems left for future research.
As to our knowledge these questions have not been answered for stabilizer codes over finite fields either.

\section{Ring-Theoretic Preliminaries}
In this short section we collect the ring-theoretic background needed later.
Let~$R$ be a finite commutative ring with identity with group of units~$R^*$.
We denote by~$\soc(R)$ the socle of the $R$-module~$R$ and by
$\rad(R)$ the Jacobson radical of~$R$.
Moreover, denote by~$\Rhat:=\text{Hom}(R,\C^*)$ the group of characters of~$R$
(i.e., group homomorphisms from $(R,\,+\,)$ to $\C^*$).
We denote the trivial character, that is the map $r\mapsto 1$ for all $r\in R$, by~$\veps_R$.
The most fundamental property of characters on~$R$ (or in fact any finite abelian group) is the orthogonality relation
\begin{equation}\label{e-RChar}
   \sum_{r\in R}\chi(r)=\left\{\begin{array}{cl}
          0,&\text{if }\chi\neq\veps_R,\\ |R|,&\text{if }\chi=\veps_R.\end{array}\right.
\end{equation}
The character group~$\Rhat$ carries an $R$-module structure via scalar multiplication
\begin{equation}\label{e-bimodule}
    (r\!\cdot\!\chi)(v)=\chi(vr) \text{ for all } r\in R\text{ and }v\in R.
\end{equation}

Later we will restrict ourselves to finite Frobenius rings.
We now summarize the crucial properties of finite commutative Frobenius rings.
Details can be found in Lam~\cite[Th.~(16.14)]{Lam99}, Lamprecht~\cite{Lamp53}, Hirano~\cite[Th.~1]{Hi97}, and
Wood~\cite[Th.~3.10, Prop.~5.1]{Wo99}.
The results are also true for non-commutative finite rings with according left and right versions (see Honold~\cite[Th.~1 and Th.~2]{Hon01}).
However, we will only need commutative rings in this paper.

\begin{theo}\label{T-Frob}
Let~$R$ be a finite commutative ring. The following are equivalent.
\begin{alphalist}
\item $\soc(R)\cong R/\rad(R)$ as $R$-modules.
\item $\soc(R)$ is a principal ideal, i.e., $\soc(R)=\alpha R$ for some $\alpha\in R$.
\item $\Rhat\cong R$ as $R$-modules.
\end{alphalist}
The ring~$R$ is called Frobenius if any (hence all) of the above hold true.
In this case there exists a character~$\chi$ such that $\Rhat= R\!\cdot\!\chi$.
Any such character is called a generating character of~$R$.
Any two generating characters~$\chi,\,\chi'$ differ by a unit, i.e., $\chi'=u\!\cdot\!\chi$ for some $u\in R^*$.
\end{theo}

The class of Frobenius rings is quite large.
For example, the integer residue rings~$\Z_N:=\Z/N\Z$, finite fields, finite chain rings are Frobenius and so are
finite group rings as well as matrix ring over Frobenius rings.
The class of Frobenius rings is closed under taking direct products.
For details see Wood~\cite[Ex.~4.4]{Wo99} and Lam~\cite[Sec.~16.B]{Lam99}.

The following properties of Frobenius rings are well-known.

\begin{rem}\label{R-FrobProp}
Let~$R$ be a finite commutative Frobenius ring. Then the following are true.
\begin{alphalist}
\item Let~$\chi$ be a character of~$R$.
      Then~$\chi$ is a generating character of~$R$ if and only if the only ideal contained in
      $\ker\chi:=\{a\in R\mid \chi(a)=1\}$ is the zero ideal; see~\cite[Cor.~3.6]{ClGo92}.
\item $R$ satisfies the double annihilator property, i.e.,
      $\ann(\ann(I))=I$ for each ideal~$I$ of~$R$,
      where $\ann(I)$ denotes the annihilator ideal of~$I$;
      see \cite[Th.~15.1]{Lam99}.
      In particular,
      \begin{equation}\label{e-ann}
         \ann(\rad(R))=\soc(R)\ \text{ and }\ \ann(\soc(R))=\rad(R),
      \end{equation}
      see \cite[Cor.~15.7]{Lam99}.
\item $R=R_1\times\ldots\times R_t$ for suitable local Frobenius rings~$R_i$; see~\cite[Th.~15.27]{Lam99}.
\item If~$R$ is local, that is, $\rad(R)$ is the unique maximal ideal of~$R$, then~$\soc(R)$ is the unique minimal
      ideal~\cite[Ex.~(3.14)]{Lam99}.
      In this case $R/\rad(R)$ is called the \emph{residue field of} $R$.
\end{alphalist}
\end{rem}

\section{Stabilizer Codes over Frobenius Rings}\label{SS-StabCode}
In this section we start with a finite commutative Frobenius ring and define the  appropriate Pauli group.
This needs some preparation in order to have the correct phases included.
Having the right notion of Pauli group will allow us to define stabilizers, stabilizer codes, and quantum stabilizer codes, all originating from this particular ring setting, in
such a way that stabilizer codes are exactly the self-orthogonal codes with respect to a certain symplectic inner product.
The outline of this section follows~\cite{NaKl12}, but differs from this reference in the details and some of the proofs (especially for Theorem~\ref{P-StabCodeSO}).

Throughout this section let~$R$ be a finite commutative Frobenius ring with cardinality $|R|=q$ and generating character~$\chi$.
We start with defining the ambient space for quantum-error-correcting codes, namely~$\C^{q^n}$, its decomposition into $n$ copies of qubits, and the associated
error operators.

Let
\[
  \cB=\{v_x\mid x\in R\}
\]
be an ON-basis of~$\C^q$.
Then an ON-basis of $\C^{q^n}=(\C^q)^{\otimes n}$ is given by
\begin{equation}\label{e-Bn}
  \cBn=\{v_x=v_{x_1}\otimes\ldots\otimes v_{x_n}\mid x=(x_1,\ldots,x_n)\in R^n\}.
\end{equation}

\begin{defi}\label{D-XZa}
For $a\in R$ define the two linear maps $X(a),\,Z(a):\C^{q}\longrightarrow\C^{q}$ via their action on the
basis~$\cB$ given by
\[
  X(a)(v_x)=v_{x+a},\quad Z(a)(v_x)=\chi(ax)v_x \text{ for all }x\in R.
\]
\end{defi}

Note that $X(a)$ is a permutation for all $a\in R$ and thus a unitary map.
Furthermore, the matrix representation of~$Z(a)$ w.r.t.~the basis~$\cB$ is given by $\text{diag}(\chi(ax))_{x\in R}$.
Since~$\chi$ takes complex values with absolute value~$1$, these maps are unitary as well.

We extend these maps to unitary maps on $\C^{q^n}$ in the usual way.
For $a=(a_1,\ldots,a_n)\in R^n$ set
\[
  X(a)=X(a_1)\otimes\ldots\otimes X(a_n),\quad Z(a)=Z(a_1)\otimes\ldots\otimes Z(a_n).
\]
Explicitly
\begin{equation}\label{e-XZ}
  X(a)(v_x)=v_{x+a}\ \text{ and }
  Z(a)(v_x)=\chi(a x)v_x \ \text{ for all }a,\,x\in R^n,
\end{equation}
where $ax=a\cdot x=\sum_{i=1}^n a_i x_i$ denotes the standard dot product in~$R^n$.
The second identity follows form
$Z(a)(v_x)=\chi(a_1x_1)v_{x_1}\otimes\ldots\otimes\chi(a_nx_n)v_{x_n}
=\big(\prod_{i=1}^n\chi(a_ix_i)\big)v_{x_1}\otimes\ldots\otimes v_{x_n}=\chi(a x)v_x$, where the last step is simply the
homomorphism property of~$\chi$.
Note that
\[
   X(a)^\ell=X(\ell a)\text{ and }Z(a)^\ell=Z(\ell a)\text{ for all }a\in R^n\text{ and }\ell\in\Z.
\]
Since the maps are unitary, we also have $X(a)^{-1}=X(a)^\dagger$ and $Z(a)^{-1}=Z(a)^\dagger$,
where~$M^\dagger$ denotes the conjugate transpose of the complex matrix $M$.
It is easy to see that
\begin{equation}\label{e-ab0}
  X(a)Z(b)=X(a')Z(b')\Longleftrightarrow (a,b)=(a',b')
\end{equation}
for all $(a,b),(a',b')\in R^{2n}$.

We will consider the group of unitary operators generated by all $X(a),\,Z(b)$ and certain scalar multiplications.
To this end we start with the set, called \emph{error basis},
\begin{equation}\label{e-En}
   \cE_n:=\{X(a)Z(b)\mid a,b\in R^n\}=\{X(a_1)Z(b_1)\otimes\ldots\otimes X(a_n)Z(b_n)\mid (a,b)\in\ R^{2n}\},
\end{equation}
which is a subset of~$\cU(q^n)$,  the unitary group on~$\C^{q^n}$.
Thanks to~\eqref{e-ab0} we have $|\cE_n|=q^{2n}$.
The elements of~$\cE_n$ adhere to the following multiplication rule.

\begin{lem}[\mbox{\cite[Prop.~4 and proof]{NaKl12}}]\label{L-MultEn}
Let $(a,b),(a',b')\in R^{2n}$ and consider the unitary maps $P=X(a)Z(b),\,P'=X(a')Z(b')$. Then
\[
   PP'=\chi(b a')X(a+a')Z(b+b')\ \text{ and }\ P'P=\chi(b'a)X(a+a')Z(b+b').
\]
As a consequence,
\[
    PP'=P'P\Longleftrightarrow \chi(b a'-b' a)=1.
\]
\end{lem}

In order to define the correct Pauli group associated with~$\cE_n$, we need the following result which can easily be verified using induction.
\begin{lem}\label{L-orderXaZb}
Let $(a,b)\in R^{2n}$. Then $\big(X(a)Z(b)\big)^{-1}=\chi(ba)X(-a)Z(-b)$ and more generally
\[
     \big(X(a)Z(b)\big)^\ell=\chi\Big(\frac{\ell(\ell-1)}{2}ba\Big)X(\ell a)Z(\ell b)\text{ for all }\ell\in\Z.
\]
\end{lem}

Now we are ready to define a parameter that will be crucial for the definition of the Pauli group.

\begin{theo}\label{T-exponent}
Let $\Char(R)=c$ and define $N:=\lcm\big\{|X(a)Z(b)|\,\big|\, (a,b)\in R^{2n}\big\}$,
where $|\cdot|$ denotes the order of the given matrix in the unitary group $\cU(q^n)$. Then
\[
   N=\left\{\begin{array}{cl}c,&\text{if $c$ is odd,}\\[.5ex] 2c,&\text{if~$c$ is even.}\end{array}\right.
\]
\end{theo}

\begin{proof}
First of all, from Lemma~\ref{L-orderXaZb} it follows that~$c$ divides $|X(e_1)Z(e_1)|$, where $e_1=(1,0,\ldots,0)$, and thus
\[
    c\mid N.
\]
(1) Let~$c$ be odd. Then $(c-1)/2$ is an integer and thus $c(c-1)/2=0$ in~$R$.
But then $(X(a)Z(b))^c=I$ by Lemma~\ref{L-orderXaZb}.
\\
(2) Let~$c$ be even.
Then $2c(2c-1)/2=0$ in~$R$ and therefore Lemma~~\ref{L-orderXaZb} implies $\big(X(a)Z(b)\big)^{2c}=I$ for all $(a,b)\in R^{2n}$.
Hence $N\mid 2c$.
It remains to find $(a,b)\in R^{2n}$ such that $\big(X(a)Z(b)\big)^c\neq I$.
Notice first that in~$R$ we have $c(c-1)/2=-c/2$.
Set $b=(1,\ldots,0)$ and $a=(\alpha,0,\ldots,0)$, where $\alpha\in R$.
Then $\big(X(a)Z(b)\big)^c=\chi((-c/2)\alpha)I$ and it suffices to show the existence of some $\alpha\in R$ such that $\chi((c/2)\alpha)\neq1$.
This can be established as follows.
Suppose $\chi((c/2)\alpha)=1$ for all $\alpha\in R$.
Then  the character $c/2\cdot\chi$ is trivial (see~\eqref{e-bimodule}).
But this contradicts the fact that~$\chi$ is a generating character, which means that $r\longmapsto r\chi$ is an $R$-isomorphism.
\end{proof}

Now we are ready to define the Pauli group.

\begin{defiprop}\label{D-Pauli}
Let $\Char(R)=c$ and~$N$ be as in Theorem~\ref{T-exponent}. Furthermore, let $\omega\in\C^*$ be a primitive $N$-th root of unity.
Then $\chi(\alpha)\in\group{\omega}:=\{1,\omega,\ldots,\omega^{N-1}\}$ for all $\alpha\in R$.
As a consequence, the subset
\[
   \cP_n:=\{\omega^\ell X(a)Z(b)\mid \ell\in\Z,\,a,b\in R^{n}\}\subseteq\cU(q^n)
\]
 is a subgroup of~$\cU(q^n)$, called the \emph{$n$-qubit Pauli group associated with the error basis} $\cE_n$.
The elements of~$\cP_n$ are called \emph{Pauli operators}.
Furthermore, the map
\[
  \Psi:\cP_n\longrightarrow R^{2n},\quad \omega^\ell X(a)Z(b)\longmapsto (a,b),
\]
is a surjective group homomorphism with $\ker\Psi=\{\omega^\ell I\mid \ell\in\Z\}$.
The latter is also the center of~$\cP_n$.
\end{defiprop}

\begin{proof}
First of all, $1=\chi(0)=\chi(c\alpha)=\chi(\alpha)^c$ implies $\chi(\alpha)\in\group{\omega}$ for all $\alpha\in R$.
Together with Lemmas~\ref{L-MultEn} and~\ref{L-orderXaZb} this establishes that~$\cP_n$ is indeed a subgroup of~$\cU(q^n)$.
The same lemmas also prove the remaining statements about~$\Psi$ and the center.
\end{proof}

The space~$\C^{q^n}$ carries a trace-inner product in the natural way and the error basis~$\cE_n$ forms an ON-basis.

\begin{rem}\label{R-TracePn}
\begin{alphalist}
\item Let $e=\omega^\ell X(a)Z(b)\in\cP_n$, where $\ell\in\Z,\,a,b\in R^n$. Then the trace of the endomorphism~$e$ is
      \[
        \Tr(e)=\left\{\begin{array}{cl}\omega^\ell q^n,&\text{if }(a,b)=(0,0),\\ 0,&\text{otherwise.}
        \end{array}\right.
      \]
      This can be seen as follows.
      If $a\neq0$, then all diagonal entries of the permutation matrix $X(a)$ are zero (see~\eqref{e-XZ})
      and thus so are those of $X(a)Z(b)$ because $Z(b)$ is a diagonal matrix. Hence the trace is zero.
      Let now $a=0$. Since $Z(b)=\text{diag}\big(\chi(xb)\mid x\in R^n\big)$, we obtain
      $\Tr\big(\omega^\ell Z(b)\big)=\omega^\ell \sum_{x\in R^n}\chi(xb)=\omega^\ell\sum_{x\in R^n}(b\chi)(x)$.
      Now the orthogonality relations~\eqref{e-RChar} and the fact that~$\chi$ is a generating character yield the desired result.
\item The same reasoning as in~(a) shows that the set~$\cE_n$ is an ON-basis of $\C^{q^n\times q^n}$
      with respect to the hermitian inner product
      \begin{equation}\label{e-TraceInner}
        \inner{P}{\hat{P}}:=\frac{1}{q^n}\Tr(P\hat{P}^\dagger).
      \end{equation}
      Indeed, for $P=X(a)Z(b)$ and $\hat{P}=X(\hat{a})Z(\hat{b})\in\cE_n$ we have
      $P\hat{P}^\dagger=X(a)Z(b)Z(\hat{b})^\dagger X(\hat{a})^\dagger=\chi(-\hat{a}(b-\hat{b}))X(a-\hat{a})Z(b-\hat{b})$
      by Lemma~\ref{L-MultEn}.
      In the language of Knill~\cite{Kn96} (see also~\cite[p.~4893]{KKKS06}), the set~$\cE_n$ forms a \emph{nice error basis}.
\end{alphalist}
\end{rem}

The following inner product will allow us to translate the noncommutativity of~$\cP_n$ to the $R^{2n}$ setting via the homomorphism~$\Psi$.

\begin{defi}\label{D-Inner}
The \emph{symplectic inner product} on $R^{2n}$ is defined as
\[
   \inners{\,\cdot\,}{\,\cdot\,}: R^{2n}\times R^{2n},\quad \inners{(a,b)}{(a',b')}=ba'-b'a.
\]
For~$C\subseteq R^{2n}$ we define $C^\perp:=\{v\in R^{2n}\mid \inners{v}{w}=0\text{ for all }w\in C\}$.
If $C$ is a submodule of~$R^{2n}$ we call $C^\perp$ the \emph{dual module}.
As usual, $C$ is called \emph{self-orthogonal} (resp.\ \emph{self-dual}) if $C\subseteq C^\perp$ (respectively, $C=C^\perp$).
\end{defi}

Since  every vector is orthogonal to itself (isotropic), every cyclic $R$-module in $R^{2n}$ (that is, a module generated by one vector) is self-orthogonal.

Note that $\inners{\,\cdot\,}{\,\cdot\,}$ is bilinear thanks to the commutativity of~$R$.
Of course, the symplectic inner product can be expressed in terms of the standard dot product (that is, matrix multiplication):
\begin{equation}\label{e-sympldot}
  \inners{(a,b)}{(a',b')}=\begin{pmatrix}a&b\end{pmatrix}\begin{pmatrix}0&-I_n\\ I_n&0\end{pmatrix}\begin{pmatrix}a'\\b'\end{pmatrix}
  \ \text{ for all }\ a,b,a',b'\in R^n.
\end{equation}

Before studying the dual codes and self-orthogonality in detail, let us record the following well-known and helpful fact about modules over Frobenius rings; see~\cite[Cor.~5]{HoLa01}.

\begin{rem}\label{R-CCperpCard}
Let $C\subseteq R^{2n}$ be a submodule. Then $|C|\!\cdot\!|C^\perp|=|R^{2n}|$.
\end{rem}

The next result is the reason for why in Definition~\ref{D-Stab}(2) below we require stabilizer codes to be $R$-submodules (instead of just additive subgroups).

\begin{prop}\label{P-ChiInner}
Let $C\subseteq R^{2n}$ be a submodule. Then
\[
  C^{\perp}=\{v\in R^{2n}\mid \chi(\inners{v}{w})=1\text{ for all }w\in C\}.
\]
\end{prop}

\begin{proof}
``$\subseteq$'' is obvious.
For ``$\supseteq$'' let $v\in R^{2n}$ such that $\chi(\inners{v}{w})=1$ for all $w\in C$.
The bilinearity of~$\inners{\cdot}{\cdot}$ implies that
$I_v:=\{\inners{v}{w}\mid w\in C\}$ is an ideal of~$R$.
By choice of~$v$ this ideal is contained in the kernel of~$\chi$.
Since~$\chi$ is a generating character, Remark~\ref{R-FrobProp}(a) implies~$I_v=\{0\}$.
All of this shows that $v\in C^\perp$.
\end{proof}

Now we are ready to define the following objects in our setting.
Recall the homomorphism~$\Psi$ and its kernel from Definition~\ref{D-Pauli}.
Later in Definition~\ref{D-Dist} we will turn to the distance of stabilizer and quantum stabilizer codes.
All of this will show that the study of stabilizer codes is easier than that of quantum stabilizer codes, yet provides
very useful information about the error-correcting quality of the latter.

\begin{defi}\label{D-Stab}
\begin{arabiclist}
\item A subgroup~$S$ of~$\cP_n$  is called a \emph{stabilizer} if
      \[
         S\text{ is abelian\quad and\quad}S\cap\ker\Psi=\{I_{q^n}\}.
      \]
\item A submodule~$C$ of~$R^{2n}$ is called a \emph{stabilizer code} if $C=\Psi(S)$ for some stabilizer~$S\leq\cP_n$.
\item A subspace~$\cQ$ of~$\C^{q^n}$ is called a \emph{quantum stabilizer code} if there exists a stabilizer $S\leq\cP_n$ such that
      \[
         \cQ=\cQ(S):=\{v\in\C^{q^n}\mid Pv=v\text{ for all }P\in S\}=\bigcap_{P\in S}\text{eig}(P;1),
      \]
      where $\text{eig}(P;1)$ denotes the eigenspace of~$P$ to eigenvalue~$1$.
      If $\dim\cQ=1$, then~$\cQ$ is also called a \emph{stabilizer state}.
\end{arabiclist}
\end{defi}

Note that by the very definition of a stabilizer we have $\Psi(S)\cong S$ (as groups) for any stabilizer~$S$.
However, for $\Psi(S)$ to be a stabilizer code, we require it to be an $R$-module.
The latter does not follow from the group structure of~$S$, which is simply due to the fact that we have no
information about the relation of $X(ra)$ and $X(a)$ for $r\in R$ and similarly for~$Z(rb)$ and~$Z(b)$.
The following example illustrates this.

\begin{exa}\label{E-AddStabCode}
Let $R$ be the field~$\F_4=\{0,1,\alpha,\alpha^2\}$ and $n=1$.
Take the generating character~$\chi$ on~$\F_4$ defined by $\chi(1)=1$ and $\chi(a)=-1$ for $a \in\{\alpha,\alpha^2\}$.
Set $S=\{I_{4},X(1),Z(1),X(1)Z(1)\}$, which is an abelian subgroup of $\cP_1\leq\cU(4)$.
Then $C:=\Psi(S)=\{(0,0),\,(1,0),\,(0,1),\,(1,1)\}$ is a subgroup of~$\F_4^2$, but not a subspace.
In fact, $C=\F_2^2$.
Note that $C\not\subseteq C^\perp$, which shows that for ``$\Rightarrow$'' in Theorem~\ref{P-StabCodeSO} below the assumption of~$C$
being $R$-linear is indeed necessary.
Using the ON-basis
\[
   \cB=\{v_0=(1,0,0,0)\T,\,v_1=(0,1,0,0)\T,\,v_\alpha=(0,0,1,0)\T,\,v_{\alpha^2}=(0,0,0,1)\T\},
\]
the unitary maps $X(1)$ and $Z(1)$ have matrix representation
\[
    X(1)=\begin{pmatrix}0&1&0&0\\1&0&0&0\\0&0&0&1\\0&0&1&0\end{pmatrix},\
    Z(1)=\begin{pmatrix}1&0&0&0\\0&1&0&0\\0&0&-1&0\\0&0&0&-1\end{pmatrix}.
\]
This results in the associated quantum stabilizer codes $\cQ(S)=\group{(1,1,0,0)\T}_{\C}$.
Checking all matrices of the Pauli group~$\cP_1$ one easily verifies that the above subgroup~$S$ is the only abelian subgroup whose associated quantum stabilizer code is $\cQ(S)$.
\end{exa}

The $R$-module structure results in the following nice characterization of stabilizer codes.
As the proof will show, this characterization requires a careful definition of the Pauli group, namely with the phases from the group~$\group{\omega}$
of $N$-th roots of unity as done in Definition~\ref{D-Pauli}.

\begin{theo}\label{P-StabCodeSO}
Let $C\subseteq R^{2n}$ be a submodule. Then
\[
  C\text{ is a stabilizer code }\Longleftrightarrow C\subseteq C^\perp.
\]
Thus, the stabilizer codes are exactly the self-orthogonal submodules with respect to the symplectic inner product.
\end{theo}

\begin{proof}
``$\Rightarrow$'' Let $C=\Psi(S)$ for some stabilizer $S\leq\cP_n$. We have to show that $C\subseteq C^\perp$.
Let $v=(a,b),\,w=(a',b')\in C$. Then $P=\Psi^{-1}(v)$ and $P'=\Psi^{-1}(w)$ are in~$S$ and therefore commute.
Write $P=\omega^\ell X(a)Z(b),\,P'=\omega^{\ell'}X(a')Z(b')$.
Then Lemma~\ref{L-MultEn} and the definition of the symplectic inner product yield $\chi\big(\inners{v}{w}\big)=1$.
Since this is true for all $w\in C$ and~$C$ is an $R$-module, Proposition~\ref{P-ChiInner} implies $v\in C^\perp$.
All of this shows $C\subseteq C^\perp$.
\\
``$\Leftarrow$''
Recall~$N$ and~$\omega$ from Definition~\ref{D-Pauli}.
Let now $C\subseteq R^{2n}$ be a self-orthogonal submodule.
Hence $ba'=b'a$ for all $(a,b),\,(a',b')\in C$.
Consider the subset
\[
     \cG=\{\omega^\ell X(a)Z(b)\mid \ell\in\Z,(a,b)\in C\}
\]
of the Pauli group~$\cP_n$.
Then Lemma~\ref{L-MultEn} implies that~$\cG$ is an abelian subgroup of~$\cP_n$.
Let $\xi$ be a nontrivial character of~$(\cG,\cdot)$ such that $\xi(\omega^\ell I)=\omega^\ell$ for all $\ell\in\Z$.
Such character~$\xi$ does indeed exist as every character of a subgroup is the restriction of a character of the entire group.
Define
\[
    S:=\{\xi\big(\chi(ab)X(-a)Z(-b)\big)X(a)Z(b)\mid (a,b)\in C\}\subseteq\cG.
\]
This is indeed a subset of~$\cG$ because~$(\chi(ab)X(a)Z(b))^N=I$, which implies that $\xi(\chi(ab)X(a)Z(b))$ is an $N$-th root of unity.
Clearly the elements of~$S$ mutually commute and $S\cap \ker\Psi=\{I\}$.
Thus it remains to show that~$S$ is a subgroup of~$\cG$, i.e., is closed under multiplication and taking inverses.
Let $(a,b),\,(a',b')\in C$. With the aid of Lemma~\ref{L-MultEn}  we compute
\begin{align*}
   &\xi\big(\chi(ab)X(-a)Z(-b)\big)X(a)Z(b)\xi\big(\chi(a'b')X(-a')Z(-b')\big)X(a')Z(b')\\
   &\hspace*{7em}=\xi\big(\chi(ab+a'b')X(-a)Z(-b)X(-a')Z(-b')\big)\chi(ba')X(a+a')Z(b+b')\\
   &\hspace*{7em}=\xi\big(\chi(ab+a'b'+ba')X(-a-a')Z(-b-b')\big)\chi(ba')X(a+a')Z(b+b')\\
   &\hspace*{7em}=\xi\big(\chi(ab+a'b'+ba')X(-a-a')Z(-b-b')\big)\xi\big(\chi(ba')I\big)X(a+a')Z(b+b')\\
   &\hspace*{7em}=\xi\big(\chi(ab+a'b'+ba'+b'a)X(-a-a')Z(-b-b')\big)X(a+a')Z(b+b')\\
   &\hspace*{7em}=\xi\big(\chi((a+a')(b+b'))X(-a-a')Z(-b-b')\big)X(a+a')Z(b+b')
\end{align*}
where in the third step we used that $\xi\big(\chi(ba')I\big)=\chi(ba')$ and in the fourth step
that  $ba'=b'a$ for all $(a,b),\,(a',b')\in C$.
All of this shows that~$S$ is closed under multiplication.
Using $(a',b')=(-a,-b)$ we also obtain closedness with respect to taking inverses.
\end{proof}

Note that by Lemma~\ref{L-orderXaZb} the stabilizer group defined above is simply
\begin{equation}\label{e-Sexpl}
    S=\{\xi(e^{-1})e\mid e=X(a)Z(b) \text{ for some }(a,b)\in C\}.
\end{equation}

One should notice that in ``$\Leftarrow$'' of the last proof we did not make use of the $R$-linearity of~$C$.
The latter is only required because stabilizer codes are linear by definition.
However, ``$\Rightarrow$'' is not true for general subgroups $\Psi(S)$ as Example~\ref{E-AddStabCode} shows.
Let us illustrate ``$\Leftarrow$'' by an example.

\begin{exa}\label{E-StabSO}
Consider the field~$\F_4=\{0,1,\alpha,\alpha^2\}$ with the generating character~$\chi$ as in Example~\ref{E-AddStabCode}.
Let~$n=1$ and $C=\im(1,\,1)\subseteq\F_4^2$, which is clearly self-orthogonal.
According to Theorem~\ref{T-exponent} and Theorem~\ref{P-StabCodeSO} we need~$\omega$ to be a fourth root of unity, that is, $\omega= i$.
Hence the group~$\cG$ in the above proof is given by
$\cG=\{i^\ell X(a)Z(a)\mid \ell\in\Z_4,\,a\in\F_4\}$.
One can check straightforwardly that~$\xi$ given by
\[
   \xi(i^\ell I)=i^\ell\text{ for }\ell\in\Z_4\ \text{ and }\
   \xi(X(\alpha^t)Z(\alpha^t))=i^{t(t+1)/2}\text{ for }t=0,1,2,
\]
induces a character of~$\cG$.
Now~\eqref{e-Sexpl} results in the stabilizer
\[
    S=\{I_4,\,X(1)Z(1),\,-iX(\alpha)Z(\alpha),\,iX(\alpha^2)Z(\alpha^2)\}
\]
satisfying $\Psi(S)=C$.
We would like to mention that in \cite[Sec.~II.B]{KKKS06} the authors use $\omega = -1$ for the definition of the Pauli group (generally,~$\omega$ is a 
primitive $p$-th root of unity, where~$p$ is the characteristic of the field).
As mentioned, including the correct phases in the definition of the Pauli group is crucial.
Indeed, if the Pauli group~$\cP_1$ were to be defined via $\omega=-1$, running the proof of \cite[Theorem 16]{KKKS06} for our example
leads to an abelian group~$S$ that is \emph{not} contained in~$\cP_1$.
\end{exa}

\begin{theo}\label{T-DimQ}
Let $S\leq\cP_n$ be a stabilizer and $\cQ=\cQ(S)\subseteq\C^{q^n}$ be the corresponding quantum stabilizer code.
Then $\dim\cQ=q^n/|S|$.
In particular, if the stabilizer code $C=\Psi(S)$ is free of dimension~$k$ over~$R$, then $\dim\cQ=q^{n-k}$.
As a special case, $\Psi(S)$ is self-dual iff~$\cQ$ is a stabilizer state (that is, one-dimensional).
\end{theo}

\begin{proof}
Set $P:=\frac{1}{|S|}\sum_{e\in S}e$.
Since~$S$ is a group we have $e'P=P$ for all $e'\in S$, and therefore $P^2=P$.
Thus~$P$ is a projection.
One easily shows that $\cQ(S)=\im P$ and hence
$\dim\cQ(S)=\dim(\im P)=\Tr(P)=q^n/|S|$, where the last step follows from
Remark~\ref{R-TracePn}(a) along with the fact that $S\cap\{\lambda I_{q^n}\mid \lambda\in \C\}=\{I_{q^n}\}$.
The second statement is clear, while the last one on self-dual codes
is a consequence of Remark~\ref{R-CCperpCard}.
\end{proof}

Now we can define the relevant distance measures.
We start with the underlying weight functions.

\begin{defi}\label{D-Weight}
\begin{alphalist}
\item The \emph{symplectic weight} of a vector $(a,b)=(a_1,\ldots,a_n,b_1,\ldots,b_n)\in R^{2n}$ is defined as
       \[
          \wts(a,b):=|\{i\mid (a_i,b_i)\neq(0,0)\}|.
       \]
\item The \emph{weight} of a Pauli operator $P\in\cP_n$, denoted by $\wt(P)$, is defined as $\wts(\Psi(P))$.
      Thus, $\wt(\omega^\ell X(a)Z(b))=\wts(a,b)$.
\end{alphalist}
\end{defi}

As for Part~(a) note that the symplectic weight of $(a_1,\ldots,a_n,b_1,\ldots,b_n)\in R^{2n}$ is the same as the Hamming weight of the vector
$(a_1,b_1,a_2,b_2,\ldots,a_n,b_n)\in (R^2)^n$, where we consider $R^2$ as the alphabet.
In Part~(b) observe that,  by definition, $\wt(\omega^\ell X(a)Z(b))= \wt(X(a)Z(b))$, that is, the phase $\omega^\ell$ is irrelevant for the weight.

For the following recall from Theorem~\ref{P-StabCodeSO} that a stabilizer code $C$ is self-orthogonal.
The relevance of the distance notions will be addressed after the definition.

\begin{defi}\label{D-Dist}
\begin{alphalist}
\item For any set $A\subseteq R^{2n}$ the  \emph{minimum distance} of~$A$ is defined as
      \[
        \ds(A)=\min\{\wts(v)\mid v\in A-\{0\}\}.
      \]
\item The \emph{relative minimum distance} of a stabilizer code $C\subseteq R^{2n}$ is defined as
      \[
        \dist_{C^\perp}(C):=\ds(C^{\perp}-C) \text{ if }C\subsetneq C^\perp\ \text{ and }\
        \dist_{C^\perp}(C):=\ds(C^{\perp}) \text{ if }C=C^\perp.
      \]
\item The \emph{minimum distance} of a quantum stabilizer code $\cQ:=\cQ(S)$, denoted by $\dist(\cQ)$, is defined as the relative minimum distance of the corresponding
      stabilizer code $\Psi(S)$.
      We say that $\cQ\subseteq\C^{q^n}$ is an $(\!(n,K,d)\!)_q$-quantum code if $\dim_{\C}(\cQ)=K$ and $\dist(\cQ)=d$.
      An $(\!(n,q^k,d)\!)_q$-quantum code is also called an $[\![n,k,d]\!]_q$-code.
\item  A quantum stabilizer code $\cQ:=\cQ(S)$ and its corresponding stabilizer code~$C=\Psi(S)$ are
        called \emph{pure} if $\ds(C^\perp)=\dist_{C^\perp}(C)$; see~\cite[p.~1373]{CRSS98}. Otherwise~$\cQ$ is called \emph{impure}.
\end{alphalist}
\end{defi}

One has to keep in mind that the notation $(\!(n,K,d)\!)_q$ does not specify the underlying Frobenius ring~$R$, but only its cardinality.
The ring has to be clear from the context.

We close this section with some background on the notions used. When $R = \F_2$, Definition~\ref{D-XZa} gives $X(0)=Z(0)=I_2$ and
\[
    X(1)=\begin{pmatrix}0&1\\1&0\end{pmatrix},\
    Z(1)=\begin{pmatrix}1&0\\0&-1\end{pmatrix}.
\]
For the standard ON basis $\cB=\{v_0=(1,0)\T,\,v_1=(0,1)\T\}$ of $\mathbb{C}^2$ we have $X(1)(v_0) = v_1$ and $X(1)(v_1) = v_0$ and thus $X(1)$ is a \emph{bit-flip error}.
On the other hand, for a vector $v = av_0 + bv_1$ we have $Z(1)(v) = av_0 - bv_1$ and thus $Z(1)$ is a \emph{phase-flip error}.
It turns out that these are the fundamental errors in quantum error-correction. In this sense $\cE_1$ or $\cE_n$ play the role of an \emph{error basis}.
Definition~\ref{D-Weight}(b) is motivated by the following facts.
Let  $C_{\cP_n}(S)$ denote the centralizer of the subgroup~$S$ in~$\cP_n$.
If $\Psi(S)$ is an $R$-module, then $\Psi(C_{\cP_n}(S) -S)=C^{\perp}-C$ thanks to Lemma~\ref{L-MultEn} and Proposition~\ref{P-ChiInner}.
It is known from the theory of quantum error correction that any error outside the difference set $C_{\cP_n}(S) -S$
is automatically detected by the corresponding quantum stabilizer code $\cQ(S)$; see
\cite[Lem.~11]{KKKS06}, \cite[Lem.~1]{CRSS98}, or \cite[Thm.~3.4]{Kn96a} for a more general approach.
As a consequence, only the minimum distance on the difference set $C^{\perp}-C$ matters.
Moreover, a quantum stabilizer code~$\cQ$ of  minimum distance $\dist(\cQ)=2t+1$ can correct any~$t$ errors; see \cite[Thm.~1]{CRSS98}.
All of this shows that $\dist_{C^\perp}(C)$ is the relevant parameter for error-correcting purposes.
Finally note that for a Pauli operator $P = w^lX(a)Z(b) = w^l(P_1\otimes\cdots\otimes P_n)$, where $P_i=X(a_i)Z(b_i)$, the weight
$\wt(P) = |\{i\mid P_i \neq I\}|$ may be regarded as the Hamming weight on~$\cP_n$ (after scaling out the phases).
Therefore the entire situation resembles the classical case that a code of minimum Hamming distance $2t+1$
can correct any~$t$ errors.
In this context it is worth mentioning that the symplectic weight disregards the phases.
This is permissible due to linearity of quantum error-correction~\cite[Thm.~3.1]{KnLa97}.

\section{Codes over Local Frobenius Rings and Their Reductions}\label{SS-RedDim}
In this section we will consider stabilizer codes over a local commutative Frobenius ring.
In this case the ring has a residue field and we relate the given stabilizer code to the induced one over the residue field.
Special attention will be paid to free stabilizer codes.

From now on let~$R$ be a finite local commutative Frobenius ring with maximal ideal $\mm$.
Then $\soc(R)$ is principal, see Remark~\ref{R-FrobProp}(d), say
\begin{equation}\label{e-soc}
   \soc(R)=\alpha R.
\end{equation}
Let $\F=R/\mm$ denote the residue field of~$R$.
By the double annihilator property~\eqref{e-ann}, the socle $\alpha R$ is an $\F$-vector space in the natural way and the map
$\alpha R\longrightarrow\F$ given by $\alpha r\longmapsto \ov{r}:=r+\mm$ is an $\F$-isomorphism.
On the other hand, the map $R\longrightarrow\F,\ r\longmapsto\ov{r}$ is clearly a ring homomorphism.

We extend the above map coordinatewise to the $\F$-isomorphism
\begin{equation}\label{e-rho}
  \rho:\alpha R^{2n}\longrightarrow\F^{2n},\quad \alpha(r_1,\ldots,r_{2n})\longmapsto (\ov{r_1},\ldots,\ov{r_{2n}}).
\end{equation}
Write $\ov{r}:= (\ov{r_1},\ldots,\ov{r_{2n}})$ for $r=(r_1,\ldots,r_{2n})$.

\begin{rem}\label{R-rhopreserve}
The map~$\rho$  preserves the support of  vectors (that is, the index set of nonzero positions)
and thus is $\wts$-preserving (and Hamming-weight-preserving).
\end{rem}

We also need the $R$-linear surjective map induced by multiplication by~$\alpha$
\[
   m_{\alpha}: R^{2n}\longrightarrow\alpha R^{2n},\quad r\longmapsto \alpha r.
\]
Thanks to~\eqref{e-ann}, its kernel is given by $\ker m_{\alpha}=\mm^{(2n)}:=\{(v_1,\ldots,v_{2n})\mid v_i\in\mm\text{ for all }i\}$.

\begin{defi}\label{D-RedF}
\begin{alphalist}
\item For a submodule $C\subseteq R^{2n}$ we define the \emph{colon module}
        \[
               (C:\alpha):=m_\alpha^{-1}(C)=\{v\in R^{2n}\mid \alpha v\in C\}.
        \]
\item For any set $X\subseteq R^{2n}$ we define the \emph{reduction} $\ov{X}=\{\ov{x}\mid x\in X\}\in\F^{2n}$, where $\ov{x}=x+\mm$. In other words,
        \[
             \ov{X}:=\rho(\alpha X)=\{\rho(\alpha x)\mid x\in X\}.
        \]
        If  $C\subseteq R^{2n}$ is a submodule, then $\alpha C$ is an $\F$-vector space (and of course an $R$-module)
        and hence $\ov{C}\subseteq\F^{2n}$ is an $\F$-subspace of~$\F^{2n}$.
\end{alphalist}
\end{defi}

The following properties are obvious.

\begin{prop}\label{P-CCbar}
Let $C\subseteq R^{2n}$ be an $R$-submodule.
Then
\begin{arabiclist}
\item $C\subseteq (C:\alpha)$ and $\alpha(C:\alpha)=C\cap \alpha R^{2n}\subseteq C$.
\item $\ov{C}=\rho(\alpha C)\cong\alpha C\cong C/(C\cap\mm^{(2n)})$,
        where the isomorphisms are $R$-linear (and the first one is also $\F$-linear).
        As a consequence,
        \[
           |\ov{C}|=\frac{|C|}{|C\cap\mm^{(2n)}|}.
        \]
\end{arabiclist}
\end{prop}

Later on we will mainly focus on free codes in~$R^{2n}$.
In that case we have the following statement about the dimensions of~$C$ and~$\ov{C}$.

\begin{cor}\label{C-DimFree}
Let $C\subseteq R^{2n}$ be a free $R$-module. Then
\[
   \dim_R(C)=\dim_{\F}(\ov{C}).
\]
\end{cor}

\begin{proof}
Let $\dim C=k$, thus $|C|=|R|^k$.
We want to show that $|C\cap\mm^{(2n)}|=|\mm^{(k)}|$ because then Proposition~\ref{P-CCbar}(2) tells us that
$|\ov{C}|=|R|^k/|\mm|^k=|\F|^k$, and thus $\dim_{\F}(\ov{C})=k$, as desired.

Let $v_1,\ldots,v_k$ be a basis of~$C$ and let $v=\sum_{i=1}^kr_i v_i\in C\cap\mm^{(2n)}$ for some $r_i\in R$.
Then~\eqref{e-ann} yields $0=\alpha v$ and the linear independence of $v_1,\ldots,v_k$ implies
$\alpha r_i=0$ for all~$i$.
Hence again by~\eqref{e-ann} we conclude $r_i\in\mm$ for all~$i$, thus $(r_1,\ldots,r_k)\in\mm^{(k)}$.
Now the linear independence of $v_1,\ldots,v_k$ implies the desired identity $|C\cap\mm^{(2n)}|=|\mm^{(k)}|$.
\end{proof}

Next we derive a convenient normal form for free codes in~$R^{2n}$ using
transformations that preserves the symplectic inner product and the symplectic weight, and thus self-orthogonality.
The normal form will also result in a normal form for the dual code.

\begin{defi}\label{D-Iso}
Let $C\subseteq R^{2n}$ and $f:C\longmapsto R^{2n}$ be $R$-linear.
Then~$f$ is called a \emph{symplectic isometry} if $\wts(a) = \wts(f(a))$ and
$\inners{a}{b} = \inners{f(a)}{f(b)}$ for all $a,b \in R^{2n}$.
Two codes $C,C'\subseteq R^{2n}$ are called \emph{symplectically isometric} if there exists a symplectic isometry
$f:C\longmapsto R^{2n}$ such that $f(C)=C'$.
\end{defi}

While general symplectic isometries will be considered in Section~\ref{SS-Iso}, the following ones will suffice for this section.

\begin{defi}\label{D-SEquiv}
\begin{alphalist}
\item For every permutation $\sigma\in S_n$ define the map $\tau_\sigma: R^{2n}\longrightarrow R^{2n}$ given by
      \[
        (a_1,\ldots,a_n,b_1,\ldots,b_n)\longmapsto (a_{\sigma(1)},\ldots,a_{\sigma(n)},b_{\sigma(1)},\ldots,b_{\sigma(n)}).
      \]
\item For every $i\in[n]:=\{1,\ldots,n\}$ we define the map $\tau_i: R^{2n}\longrightarrow R^{2n}$ given by
      \[
        (a_1,\ldots,a_n,b_1,\ldots,b_n)\longmapsto(a_1,\ldots,a_{i-1},b_i,a_{i+1},\ldots,a_n,b_1,\ldots,b_{i-1},-a_i,b_{i+1},\ldots,b_n).
      \]
\end{alphalist}
\end{defi}

\begin{prop}\label{P-SympPres}
The maps $\tau_\sigma, \sigma\in S_n$, and $\tau_i,\,i\in[n],$ are symplectic isometries.
\end{prop}

\begin{proof}
It is clear that both types of maps are $R$-linear and preserve the symplectic weight.
Furthermore, it is obvious that $\tau_\sigma$ preserves the symplectic inner product.
Let now $i\in[n]$ and consider~$\tau_i$. Then for $(a,b),\,(a',b')\in R^{2n}$ we compute
\[
  \inners{\tau_i(a,b)}{\tau_i(a',b')}=\sum_{j\neq i} b_ja'_j-a_ib'_i-\sum_{j\neq i} a_jb'_j+b_ia'_i
   =\sum_{j=1}^n b_ja'_j-\sum_{j=1}^n a_jb'_j=\inners{(a,b)}{(a',b')},
\]
which proves the desired result.
\end{proof}

It is clear that if $C\subseteq R^{2n}$ is self-orthogonal then so is every
symplectically isometric code.
Recall from Proposition~\ref{P-StabCodeSO} that the stabilizer codes are exactly the self-orthogonal submodules of~$R^{2n}$.

\begin{theo}\label{T-FreeCode}
Let $C\subseteq R^{2n}$ be a free stabilizer code of dimension~$k$.
Then~$C$ is symplectically isometric to a free code~$C'$ of the form
\begin{equation}\label{e-CpMat}
  C'=\im\begin{pmatrix}I_k&M&N_1&N_2\end{pmatrix}\subseteq R^{2n},
\end{equation}
where $M\in R^{k\times(n-k)},\, N_1\in R^{k\times k},\,N_2\in R^{k\times(n-k)}$ such that $N_1+N_2M\T\in R^{k\times k}$
is a symmetric matrix.
Furthermore, if~$C'$ is as in~\eqref{e-CpMat}, then~$C'$ is self-orthogonal iff $N_1+N_2M\T$ is symmetric.
\end{theo}

The proof will show that we only need the symplectic isometries~$\tau_\sigma$ and~$\tau_i$ in order to transform~$C$ to the code~$C'$.

\begin{proof}
Let us first verify the very last statement.
Using~\eqref{e-sympldot}, we see that the code~$C'$ is self-orthogonal iff
\[
  \begin{pmatrix}I_k&M&N_1&N_2\end{pmatrix}\begin{pmatrix}-N_1\T\\-N_2\T\\I_k\\M\T\end{pmatrix}=0.
\]
But this identity is equivalent to $N_1+N_2M\T=(N_1+N_2M\T)\T$, which is exactly the stated symmetry.
This proves the last statement.

Now we turn to the first part.
Write $C=\im G$ for some matrix $G\in R^{k\times2n}$, thus the rows of~$G$ form a basis of~$C$.
We induct on~$k$.
Let~$k=1$.
Since~$R$ is local Frobenius, there exists a unit in the first (and only) row of~$G$
for otherwise it would be annihilated by~$\alpha$ (the generator of the socle),  contradicting the fact that it is a basis vector.
If the unit, say~$u$, is among the first~$n$ entries, we use an automorphism of the form~$\tau_\sigma$ to bring~$u$
into position~$(1,1)$.
If the unit~$u$ is among the last~$n$ entries, we use an automorphism of the type $\tau_i$ to bring it into one of
the first~$n$ positions and proceed with an automorphism~$\tau_\sigma$ to bring it to position~$(1,1)$.
In either case we arrive at a symplectically isometric code of the form
\[
  \im(u,g'_2,\ldots,g'_{2n})=\im(1,u^{-1}g'_1,\ldots,u^{-1}g'_{2n}),
\]
which is what we wanted.

For the induction step we proceed with the first row of~$G$ as above.
Using also row operations (which don't change the code at all),
we arrive at a symplectically isometric code of the form
\[
    \im\left(\!\!\begin{array}{c|ccc|ccc}
      1&*&\ldots&*&*&\ldots&*\\ \hline 0& & & & & & \\
      \vdots& &G_1& & &G_2&\\0& & & & & &\end{array}\!\!\right)
\]
For the next step we need to proceed with caution. We show the general case.
Assume that we arrived at a code with generator matrix

\begin{align*}
    \hat{G}:=&\left(\!\!\begin{array}{ccc|ccc|ccc|ccc}
      1& & & & & & & & & & &\\ &\ddots & &&*&&&*& & &*& \\ & &1& & & & & & & & & \\ \hline & & & & & & & & & & &\\
      &0& & &G_1& & &\ G_2\ & & & G_3 & \\ & & & & & & & & & & &\end{array}\!\!\right)
     \hspace*{-2em}\begin{array}{l}
       \left.\begin{array}{l}\phantom{1} \\ \phantom{\ } \\ \ \end{array}\right\}r \\[3.4ex]
       \left.\begin{array}{l}\phantom{1} \\ \phantom{\ } \\ \ \end{array}\right\}k-r
     \end{array}
      \\[-1ex]
    &\mbox{}\hspace*{1em} \underbrace{\hspace*{4em}}_{r}\ \underbrace{\hspace*{3.7em}}_{n-r}\ \underbrace{\hspace*{4em}}_{r}\
     \underbrace{\hspace*{3.7em}}_{n-r}
\end{align*}
Now we continue with the matrix $(G_1,G_2,G_3)$.
Again, its first row has to contain a unit.
\\
\underline{Case 1:} The first row of~$G_1$ contains a unit.
Then we may use an automorphism~$\tau_{\sigma}$ to bring the unit to position~$(1,1)$ of~$(G_1,G_2,G_3)$.
\\
\underline{Case 2:} The first row of~$G_3$ contains a unit.
Then we may use an automorphism~$\tau_i$ to bring it to the first row of~$G_1$ and then proceed as in Case~1.
\\
\underline{Case 3:} Suppose the first rows of~$G_1$ and~$G_3$ contain no unit.
Hence all their entries are in~$\mm$.
Then there must be a unit in the first
row of~$G_2$, say $(G_2)_{1,l}=u\in R^*$, where $l\in[r]$.
But then the symplectic inner product of the $l^{\rm th}$ and $(r+1)^{\rm st}$ row of~$\hat{G}$ results in $s-u$
for some $s\in\mm$.
Since this is clearly not zero, we arrive at a contradiction to the self-orthogonality of~$C$, and thus this case cannot occur.

This shows that with a symplectic isometry we obtain a unit at position $(r+1,r+1)$ of the matrix~$\hat{G}$.
But then, using row operations, we can easily derive a matrix $\tilde{G}$ of the form as for~$\hat{G}$, but with the identity
of size $(r+1)\times(r+1)$.
Proceeding in this way we arrive at the desired form.
\end{proof}

Now we easily obtain

\begin{prop}\label{P-FreeDual}
Let $C\subseteq R^{2n}$ be a free stabilizer code of the form
\[
  C=\im\begin{pmatrix}I_k&M&N_1&N_2\end{pmatrix}\leq R^{2n},
\]
for some matrices $M\in R^{k\times(n-k)},\, N_1\in R^{k\times k},\,N_2\in R^{k\times(n-k)}$ such that $N_1+N_2M\T$ is symmetric.
Then the symplectic dual~$C^\perp$ is free with dimension $2n-\dim C$ and is given by
\[
   C^\perp =\im\begin{pmatrix}I_k&0&N_1\T&0\\0&I_{n-k}&N_2\T&0\\0&0&M\T&-I_{n-k}\end{pmatrix}
      =\im\begin{pmatrix}I_k&M&N_1&N_2\\0&I_{n-k}&N_2\T&0\\0&0&M\T&-I_{n-k}\end{pmatrix}.
\]
\end{prop}

\begin{proof}
Using $N_1+N_2M\T-N_1\T-MN_2\T=0$, one easily checks that the two matrices are obtained via the obvious elementary row operations.
Thus it suffices to consider the first matrix.
With the aid of~\eqref{e-sympldot} one verifies that its row space is contained in~$C^\perp$. The converse follows from
$|C^\perp|=q^{2n}/q^k=q^{2n-k}$, which is a consequence of~$R$ being Frobenius; see Remark~\ref{R-CCperpCard}.
The latter is indeed the cardinality of the code generated by the given matrix.
\end{proof}

Note that for any code $C\subseteq R^{2n}$ we have $\ov{C^\perp}\subseteq\ov{C}^\perp$.
The converse is not true in general. Take for instance, $C=\im(2,2)\subseteq\Z_4^2$ (hence $n=1$).
Then $C^\perp=\im\Smallfourmat{1}{1}{0}{2}$ and thus $\ov{C^{\perp}}=\im(1,1)$.
Furthermore, $\ov{C}=\{0\}$, so $\ov{C}^{\perp}=\F_2^2$.
As a consequence, $\ov{C^{\perp}}\subsetneq\ov{C}^\perp$.

However, if~$C$ is free we have the following result which is now immediate with Corollary~\ref{C-DimFree} and Proposition~\ref{P-FreeDual}.

\begin{cor}\label{C-CfreeCbar}
Let $C=\im G$ and $C^{\perp}=\im H$, where $G\in R^{k\times 2n}$ and $H\in R^{(2n-k)\times 2n}$ are as in
Proposition~\ref{P-FreeDual} ($H$ may be any of the two given matrices). Then
$\ov{C}=\im \ov{G}$ and $\ov{C^\perp}=\im\ov{H}$, hence we have normal forms for the reduced codes as well.
In particular, $\ov{C^\perp}=\ov{C}^\perp$.
\end{cor}

We close this section with the following direct consequence of Proposition~\ref{P-FreeDual} on the distance of~$C^\perp$.

\begin{rem}\label{R-dsCperp}
If~$C\subseteq R^{2n}$ is a free stabilizer code of dimension~$k$, then $\ds(C^\perp)\leq k$.
\end{rem}

Be aware that~$C^\perp$ is in general not self-orthogonal (unless it is self-dual) and thus not a stabilizer code.
Moreover, as explained at the end of Section~\ref{SS-StabCode}, the relevant parameter for measuring the error-correcting quality
of the associated quantum stabilizer code is the relative distance $\dist_{C^\perp}(C)$ and not $\ds(C^\perp)$.
We will turn to $\dist_{C^\perp}(C)$ in the next section.

\section{On the Distance of a Stabilizer Code and its Reduction}\label{SS-Dist}

In this section we compare the distance and relative distance of a stabilizer code to those of its reduction.
We will see that for free stabilizer codes the relative distance is always at most the relative distance of its
reduction (and often equal).

We start with the distance $\ds(C)$, which is easier to deal with.
The following result was proven in~\cite[Thm.~4.2]{NoSa00} for linear codes over chain rings along with the Hamming distance.
For local Frobenius rings we need a slightly more involved argument for the proof than for chain rings.

\begin{theo}\label{T-DistEst}
Let $C \subseteq R^{2n}$ such that $\ov{C}\neq0$. Then $\ds(C) = \ds\big(\ov{(C:\alpha)}\big) \leq \ds(\ov{C}).$
\end{theo}

\begin{proof}
The assumption $\ov{C}\neq0$ means that $C$ is not contained in $\mm^{(2n)}$, and thus $(C:\alpha) \nsubseteq \mm^{(2n)}$, which in turn yields
$\alpha(C:\alpha) \neq 0$.
Now the second containment in Proposition~\ref{P-CCbar}(1) implies $\ds(C) \leq \ds(\alpha(C:\alpha))$.
Making use of the $\wts$-preserving property of the isomorphism $\rho$ we conclude
\[
  \ds(C) \leq \ds\big(\alpha(C:\alpha)\big) = \ds\big(\rho\left(\alpha(C:\alpha)\right)\big) = \ds\big(\overline{(C:\alpha)}\big).
\]
For the reverse inequality, we have to show that $\ds(\alpha(C:\alpha)) \leq \ds(C)$.
We will do so by showing that
\begin{equation}\label{e-lambdav}
  \text{for every $v \in C-\{0\}$ there exists $\lambda\in R$ such that $0 \neq \lambda v\in\alpha(C:\alpha)$.}
\end{equation}
Since $0<\wts(\lambda v) \leq \wts(v)$, this proves $\ds(\alpha(C:\alpha)) \leq \ds(C)$.

Let $v\in C-\{0\}$ and set $v^{(0)} := v$.
Write $\mm = (z_1, \ldots, z_t)$.
Note that $v^{(0)}\in \alpha(C:\alpha)$ iff $v^{(0)} = \alpha w$ for some $w\in R^n$.
If $v^{(0)}\in\alpha(C:\alpha)$, we are done.
Otherwise, the entries of~$v^{(0)}$ are not all in $\alpha R=\ann(\mm)$.
Thus there exists a minimal index~$s$ such that $z_sv^{(0)}\neq0$.
Let~$k_s$ be maximal such that $v^{(1)} := z_s^{k_s}v^{(0)} \neq  0$.
Maximality implies $z_sv^{(1)} = 0$.

Now we proceed with~$v^{(1)}$.
If $v^{(1)}\in \alpha(C:\alpha)$, we are done.
Otherwise, there exists a minimal index $s'$ such that $z_{s'}v^{(1)}\neq0$.
Note that $s'>s$.
Let again $k_{s'}$ be maximal such that $v^{(2)}:= z_{s'}^{k_{s'}}v^{(1)} \neq 0$.
Hence $z_{s'}v^{(2)} = 0$.

Proceeding in this way, we eventually obtain a vector $v^{(t)}=\lambda v$ for some $\lambda\in\mm$ such that
$v^{(t)}\in \alpha(C:\alpha)-\{0\}$.
This establishes~\eqref{e-lambdav} and thus $\ds(\alpha(C:\alpha)) \leq \ds(C)$, and all of this proves $\ds(C)=\ds\big(\ov{(C:\alpha)}\big)$.

Finally, the containment $C\subseteq (C:\alpha)$ implies $\ov{C}\subseteq \ov{(C:\alpha)}$, and since
$\ov{C} \neq 0$ we conclude $\ds\big(\ov{(C:\alpha)}\big) \leq \ds(\ov{C}).$
\end{proof}

In general we do not have equality $\ds(C)=\ds(\ov{C})$.
In fact, the difference can be arbitrarily large.

\begin{exa}\label{E-dsCdsCbar}
This is an adaptation of \cite[Rem.~4.8(ii)]{NoSa00} to the symplectic case.
Consider the code $C\subseteq R^{2n}$ given as
\[
    C= \im\left(\begin{array}{c|c} {\mathbf 1}&{\mathbf 1}\\ \alpha I_n&\alpha I_n\end{array}\right),\ \text{ where }{\mathbf 1}=(1,\ldots,1)\in R^n.
\]
Then~$C$ is indeed self-orthogonal and clearly $\ds(C)=1$, whereas $\ds(\ov{C})=n$.
\end{exa}

However, as we show next for free stabilizer code the inequality in Theorem~\ref{T-DistEst} becomes an equality.
Recall that a generator matrix $G\in R^{k\times 2n}$ is called \emph{systematic} if it contains the identity $I_k$ as a $k\times k$-submatrix.

\begin{theo}\label{T-DistFree}
Let $C\subseteq R^{2n}$ be a free code with a systematic generator matrix, for instance, $C$ or~$C^{\perp}$ is a free stabilizer code (see Theorem~\ref{T-FreeCode} and
Proposition~\ref{P-FreeDual}).
Then $C\cap \alpha R^{2n} = \alpha C$. As a consequence $\ov{(C:\alpha)} = \ov{C}$ and $\ds(C) = \ds(\ov{C})$.
\end{theo}

\begin{proof}
Obviously $\alpha C \subseteq C \cap \alpha R^{2n}$.
For the converse we may assume without loss of generality that $C=\im\begin{pmatrix}I_k&M\end{pmatrix}$ for some $M\in R^{k\times(2n-k)}$.
Let $v \in C \cap \alpha R^{2n}$.
Then there exists $x \in R^k$ such that $v = x(I_k \mid M) = (x \mid xM) \in \alpha R^{2n}$.
But then $x=\alpha y$ for some $y\in R^k$ and $v=\alpha y(I_k \mid M)\in\alpha C$.
This establishes $C\cap \alpha R^{2n} = \alpha C$.
As for the last part, we obtain from Proposition~\ref{P-CCbar}(1) and~(2)
\[
  \ov{(C:\alpha)}=\rho\big(\alpha(C:\alpha)\big)=\rho(C\cap\alpha R^{2n})=\rho(\alpha C)=\ov{C}\subseteq\ov{(C:\alpha)},
\]
and conclude the proof with the aid of Theorem~\ref{T-DistEst}.
\end{proof}

In~\cite{NoSa00} the authors prove the above result for (not necessarily self-orthogonal) free codes over chain rings with the Hamming weight.
Hence their result may be regarded as the Hamming analogue of Theorem~\ref{T-DistFree}.

We now turn to the relative minimum distance $\dist_{C^\perp}(C)$, which is the relevant parameter for the error-correcting quality of quantum stabilizer codes; see Definition~\ref{D-Dist}(b).
We aim at relating the relative minimum distances $\dist_{C^\perp}(C)$ and $\dist_{\overline{C}^\perp}(\ov{C})$.

\begin{theo}\label{T-DistC}
Let $C\subseteq R^{2n}$ be a free stabilizer code.
Then $\rho^{-1}( \ov{C}^{\perp} - \ov{C})\subseteq C^{\perp} - C$.
As a consequence, $\dist_{C^\perp}(C)\leq\dist_{\overline{C}^\perp}(\ov{C})$.
Moreover,  equality holds if~$C$ is self-dual.
\end{theo}

\begin{proof}
Let $\ov{x}\in\ov{C}^\perp - \ov{C}$.
By Corollary~\ref{C-CfreeCbar} there exists $y \in C^{\perp}$ such that $\ov{y} =\ov{x}$.
Then $\rho^{-1}(\ov{x})=\rho^{-1}(\ov{y})=\alpha y\in C^\perp$.
Suppose $\alpha y\in C$. Then $y\in (C:\alpha)$, and thus $\ov{x}=\ov{y}\in\ov{(C:\alpha)}$.
But the latter is~$\ov{C}$ thanks to Theorem~\ref{T-DistFree}, and we arrive at a contradiction.
Thus $\rho^{-1}(\ov{x})=\alpha y\not\in C$.
This shows the containment $\rho^{-1}( \ov{C}^{\perp} - \ov{C})\subseteq C^{\perp} - C$.
Now  $\dist_{C^\perp}(C)\leq\dist_{\small{\overline{C}^\perp}}(\ov{C})$ follows from Remark~\ref{R-rhopreserve}.
Finally, if~$C$ is self-dual then so is $\ov{C}$, and Definition~\ref{D-Dist}(b), Corollary~\ref{C-CfreeCbar}
and Theorem~\ref{T-DistFree} yield  $\dist_{C^\perp}(C)=\ds(C^\perp)=\ds(\ov{C^\perp})=\ds(\ov{C}^\perp)=\dist_{\overline{C}^\perp}(\ov{C})$.
\end{proof}

Note that in the above statement we even have
$\rho^{-1}( \ov{C}^{\perp} - \ov{C})\subseteq \alpha C^{\perp} - C=\alpha C^{\perp}-(C\cap \alpha R^{2n})=\alpha C^{\perp}-\alpha C$.
Since $|\alpha C^{\perp}-\alpha C|=|\ov{C}^{\perp} - \ov{C}|$, we actually have equality and therefore
$\dist_{\small{\overline{C}^\perp}}(\ov{C})=\ds(\alpha C^{\perp}-\alpha C)$ since~$\rho$ is $\wts$-preserving.

\begin{conj}\label{C-RelDist}
$\dist_{\small{\overline{C}^\perp}}(\ov{C})=\dist_{C^\perp}(C)$ for any free stabilizer code $C\subseteq R^{2n}$.
\end{conj}

In other words, we believe that free stabilizer codes over local Frobenius rings do not underperform stabilizer codes over fields.
One may note that the difference set $C^\perp-C$ is much larger than $\ov{C}^\perp-\ov{C}$, yet is conjectured to have the same minimum distance.
Recall from Section~\ref{SS-StabCode} that $\cQ(C)$ is a quantum stabilizer code for the quantum state space $\C^{q^n}$, where $q=|R|$,
whereas  $\cQ(\ov{C})$ is a quantum stabilizer code for the (much smaller) quantum state space $\C^{\hat{q}^n}$, where $\hat{q}=|\F|$.
Accordingly, the relative distances provide information about the correctable errors in quite different Pauli groups; see the paragraph at the end of Section~\ref{SS-StabCode}.

We strongly believe that the conjecture is true, but are unfortunately not able to
provide a proof at this point (and we do not even have a counterexample for non-free codes).
The difficulties related to a proof will be illustrated by some examples in the next section.
We have the following partial result.

\begin{prop}\label{P-Partial}
Let $C\subseteq R^{2n}$ be a free stabilizer code and suppose $\ds(\ov{C}^{\perp})=\dist_{\small{\overline{C}^\perp}}(\ov{C})$, i.e., the reduced code~$\ov{C}$ is pure. Then
$\dist_{C^\perp}(C)=\dist_{\small{\overline{C}^\perp}}(\ov{C})$.
\end{prop}

\begin{proof}
Using Theorem~\ref{T-DistFree} we conclude
$\dist_{\small{\overline{C}^\perp}}(\ov{C})=\ds(\ov{C}^{\perp})=\ds(C^{\perp})\leq\dist_{C^\perp}(C)$.
\end{proof}

Note that the above assumption $\ds(\ov{C}^{\perp})=\dist_{\small{\overline{C}^\perp}}(\ov{C})$ means that at least one nonzero codeword of minimum symplectic weight in $\ov{C}^{\perp}$
does not lie in~$\ov{C}$.
This assumption is actually most often satisfied.
Indeed, consider the generator matrices of~$C$ and~$C^{\perp}$ in Proposition~\ref{P-FreeDual}.
Their reductions yield generator matrices for~$\ov{C}$ and $\ov{C}^{\perp}$; see Corollary~\ref{C-CfreeCbar}.
This shows that~$\ov{C}^{\perp}$ always contains some very low weight codewords.
The following is a special instance of this line of reasoning.

\begin{cor}\label{C-GoodFree}
Let~$C$ be a free stabilizer code such that $\ds(C)>\dim(C)$.
Then $\dist_{C^\perp}(C)=\dist_{\small{\overline{C}^\perp}}(\ov{C})$.
\end{cor}

\begin{proof}
The assumption implies $\ds(\ov{C})>\dim(\ov{C})$; see Corollary~\ref{C-DimFree} and Theorem~\ref{T-DistFree}.
Thus Remark~\ref{R-dsCperp} applied to~$\ov{C}$ yields $\dist_{\small{\overline{C}^\perp}}(\ov{C})=\ds(\ov{C}^{\perp})$ and Proposition~\ref{P-Partial}
concludes the proof.
\end{proof}

\section{Examples for Stabilizer Codes and Their Reductions}\label{SS-Exa}
In this section we present some examples on the relative minimum distance of a code and its reduction.
We close with some comments and comparisons to the literature.

We start with an example of a free stabilizer code~$C$ for which $\ov{C}$ is not pure, yet $\dist_{C^\perp}(C)=\dist_{\small{\overline{C}^\perp}}(\ov{C})$.
This shows that in Proposition~\ref{P-Partial} the assumption that~$\ov{C}$ be pure is not necessary.

\begin{exa}\label{E-FreeCodeNonPure}
Let $R=\Z_4,\,n=7$, and $C=\im G$, where
\[
  G=\left(\!\!\begin{array}{ccccccc|ccccccc}1&0&0&0&0&0&1&1&0&0&0&0&0&0\\0&1&0&0&0&0&3&3&2&3&2&0&1&1\\0&0&1&0&0&0&0&1&2&2&3&3&3&3\\
   0&0&0&1&0&0&2&1&3&1&3&0&2&3\\0&0&0&0&1&0&1&2&3&2&3&1&3&2\\0&0&0&0&0&1&2&0&3&1&0&3&2&0\end{array}\!\!\right)\in\Z_4^{6\times 14}.
\]
Thus~$C$ is a free code of dimension~$6$. It is self-orthogonal with dual code $C^\perp=\im H$, where
\[
  H=\left(\!\!\begin{array}{ccccccc|ccccccc}1&0&0&0&0&0&1&1&0&0&0&0&0&0\\0&1&0&0&0&0&3&3&2&3&2&0&1&1\\0&0&1&0&0&0&0&1&2&2&3&3&3&3\\
   0&0&0&1&0&0&2&1&3&1&3&0&2&3\\0&0&0&0&1&0&1&2&3&2&3&1&3&2\\0&0&0&0&0&1&2&0&3&1&0&3&2&0\\0&0&0&0&0&0&1&0&1&3&3&2&0&0\\
   0&0&0&0&0&0&0&1&3&0&2&1&2&3\end{array}\!\!\right).
\]
One checks straightforwardly that $\ds(C)=\ds(C^{\perp})=2$,
and all codewords of symplectic weight~$2$ in~$C^{\perp}$ are actually in~$C$ (namely, the nonzero multiples of the first row of~$G$).
Since $(0,1,1,0,0,0,1,0)H=(0, 1, 1, 0, 0, 0, 0, 0, 1, 0, 0, 1, 0, 0)$ is a codeword in~$C^\perp-C$ of weight~$3$ we have $\dist_{C^\perp}(C)=3$.
Let us now consider the reduced codes $\ov{C},\,\ov{C^\perp}\subseteq\F_2^{14}$.
Generator matrices~$\ov{G}$ and~$\ov{H}$ are obtained by multiplying~$G$ and~$H$ by~$2$ modulo~$4$ and then replacing all entries equal to~$2$ by~$1$.
It turns out that $\ds(\ov{C})=2=\ds(\ov{C^\perp})=2$, and the only codeword of symplectic weight~$2$ is the first row of~$\ov{G}$, thus in~$\ov{C}$.
Since the reduction of the above given codeword of weight~$3$ also has weight~$3$, we obtain
$\dist_{\small{\overline{C}^\perp}}(\ov{C})=3=\dist_{C^\perp}(C)$.
\end{exa}

The next example illustrates the situation where $\dist_{\small{\overline{C}^\perp}}(\ov{C})=\dist_{C^\perp}(C)$, but where the nonzero codewords of lowest weight
in $\ov{C}^\perp-\ov{C}$ are not the reduction of lowest weight codewords in $C^\perp-C$.

\begin{exa}\label{E-FreeSelfOrth}
Let $R=\Z_8$ and $n=3$. Consider $C=\im G,\, C^{\perp}=\im H$, where
\[
   G=\left(\!\!\begin{array}{ccc|ccc}1&0&0&1&2&2\\0&1&4&2&1&2\end{array}\!\!\!\right),\quad
   H=\left(\!\!\begin{array}{ccc|ccc}1&0&0&1&2&2\\0&1&4&2&1&2\\0&0&1&2&2&0\\0&0&0&0&4&7\end{array}\!\!\!\right).
\]
Clearly,~$C$ is self-orthogonal. In this case $\dist_{C^\perp}(C)=1$, and all weight-1-codewords in $C^\perp-C$ have all their entries in the maximal ideal $(2)$.
Hence their reductions are the zero word.
Yet, $\overline{C}^\perp-\ov{C}$ contains~$3$ codewords of weight~$1$, and thus $\dist_{\small{\overline{C}^\perp}}(\ov{C})=\dist_{C^\perp}(C)$.
They are thus the reduction of higher weight codewords in $C^\perp-C$.
\end{exa}

The above examples resemble all other examples that we have studied. They have led to Conjecture~\ref{C-RelDist}.

We now briefly consider non-free stabilizer codes.
In this case one has to carefully distinguish between $\ov{C}^\perp$ and $\ov{C^\perp}$.
In the following example we obtain $\dist_{C^\perp}(C)=\dist_{\small{\overline{C}^\perp}}(\ov{C})<\ds(\ov{C^\perp}-\ov{C})$.

\begin{exa}\label{E-Z8NonFree}
Let $R = \Z_8$ and let~$C$ be the stabilizer code given by the row space of the matrix
\[
	G:=\left(
		\begin{array}{ccccc|ccccc}
         1&0&0&0&3&0&0&2&3&0\\0&1&0&0&3&0&0&7&7&0\\0&0&2&0&0&6&0&0&0&2\\
         0&0&0&2&0&6&6&0&0&0\end{array}\right).
\]
One easily checks that~$C$ is indeed self-orthogonal. Clearly $C$ is not free.
Moreover, $|C| = 8^3\cdot 2$ and thus $|C^{\perp}|=8^6\cdot4$. The dual code is generated by
\[
  H:=\left(\begin{array}{ccccc|ccccc}
       1&0&0&0&3&0&0&2&3&0\\0&1&0&0&3&0&0&7&7&0\\0&0&1&0&0&7&4&0&0&1\\0&0&0&1&0&7&3&0&0&4\\0&0&0&0&1&0&0&1&4&0\\
       0&0&0&0&0&1&1&0&0&5\\0&0&0&0&0&0&0&4&0&0\\0&0&0&0&0&0&0&0&4&0\end{array}\right).
\]
The last row of~$H$ shows that $\dist_{C^\perp}(C)=1$.
As for the reductions we have $\ov{C}=\im G',\ \ov{C^\perp}=\im H'$, where
\[
  G'=\left(
		\begin{array}{ccccc|ccccc}
         1&0&0&0&1&0&0&0&1&0\\0&1&0&0&1&0&0&1&1&0\end{array}\right),\
  H'=\left(
		\begin{array}{ccccc|ccccc}
         1&0&0&0&1&0&0&0&1&0\\0&1&0&0&1&0&0&1&1&0\\0&0&1&0&0&1&0&0&0&1\\0&0&0&1&0&1&1&0&0&0\\0&0&0&0&1&0&0&1&0&0\\0&0&0&0&0&1&1&0&0&1
         \end{array}\right).
\]
It follows that $\ds(\ov{C^\perp}-\ov{C})=2$.
On the other hand, the dual code $\ov{C}^\perp$ certainly contains the vector $(0,0,0,0,0,0,0,1,0,0)$ (because it is symplectically orthogonal to each row of~$G'$), and thus
$\dist_{\small{\overline{C}^\perp}}(\ov{C})=\ds(\ov{C}^\perp-\ov{C})=1=\dist_{C^\perp}(C)$.
\end{exa}

\begin{rem}\label{L-NK}
\begin{alphalist}
\item In \cite{NaKl12} the authors use as the relative minimum distance of the reduction the quantity $\ds{(\ov{C^{\perp} -C})}$. In \cite[Thm.~16]{NaKl12} they show
       \begin{equation}\label{E-DistNaKL}
       \dist_{\small{C^{\perp}}}(C) \leq \ds{(\ov{C^{\perp} -C})}
       \end{equation}
       for free stabilizer codes over chain rings.
       Note that if~$C$ is a free stabilizer code then $\ov{C}^{\perp}-\ov{C} \subseteq \ov{C^{\perp}-C}$ and thus $\ds(\ov{C^{\perp}-C})\leq\dist_{\small{\overline{C}^\perp}}(\ov{C})$.
       This means that the authors of~\cite{NaKl12} compare the relative distance of a free code to a potentially smaller quantity than we do.
      We believe that a comparison of stabilizer codes over a ring versus a field should relate the relative distances of these codes, which in this case means comparing $\dist_{C^\perp}(C)$
      and $\dist_{\small{\overline{C}^\perp}}(\ov{C})$.
\item As Example \ref{E-Z8NonFree} shows, the freeness of the code is necessary for the inequality $ \ds(\ov{C^{\perp}-C})\leq\dist_{\small{\overline{C}^\perp}}(\ov{C})$ to hold.
In fact, for non-free codes this inequality may even turn around as seen in Example~\ref{E-Z8NonFree}.
      This suggests that $\ds(\ov{C^{\perp}-C})$ behaves erratically.
       Furthermore, the same example provides a case for which Inequality~\eqref{E-DistNaKL} is strict (it can be verified that $\ds(\ov{C^{\perp}-C})=2$),
       yet $\dist_{\small{C^{\perp}}}(C) = \dist_{\small{\overline{C}^\perp}}(\ov{C})$.
\end{alphalist}
\end{rem}

\section{Isometries of Stabilizer Codes}\label{SS-Iso}

In Definition~\ref{D-Iso} we introduced symplectic isometries.
As  always in coding theory, weight functions and weight-preserving linear maps   give rise to the question as to what the isometries between two codes are.
Notice that in this particular case isometries also preserve the symplectic inner product.
Therefore one should pay special attention to symplectic isometries between self-orthogonal codes.
In this section we briefly report on some basic observations for symplectic isometries.
They appear to be new also for codes over finite fields. We close with some interesting open problems.

As explained right after Definition~\ref{D-Weight}, the symplectic weight can be regarded as a Hamming weight over the alphabet~$R^2$.
In this section we will make use of this connection.
To this end consider the coordinate change
\[
  \gamma: R^{2n}\longrightarrow (R^2)^n,\quad  (a_1,\ldots,a_n\mid b_1,\ldots,b_n)\longmapsto (a_1,b_1\mid a_2,b_2\mid \ldots\mid a_n,b_n).
\]
For any~$R$-linear $f:R^{2n}\longrightarrow R^{2n}$ we define $\hat{f}:=\gamma\circ f\circ\gamma^{-1}: (R^2)^n\longrightarrow(R^2)^n$.
Thus we have the commutative diagram
\begin{equation}\label{e-Diagram}
\begin{array}{l}
   \begin{xy}
   (0,0)*+{(R^2)^n}="a"; (20,0)*+{(R^2)^n}="b";%
   (0,20)*+{R^{2n}}="c"; (20,20)*+{R^{2n}}="d";%
   {\ar "a";"b"}?*!/_2mm/{\hat{f}};
   {\ar "c";"d"}?*!/_2mm/{f};%
   {\ar "c";"a"}?*!/_2mm/{\gamma};
   {\ar "d";"b"}?*!/_2mm/{\gamma};
   \end{xy}
\end{array}
\end{equation}
Let~$\wtH$ be the Hamming weight of vectors with entries in the alphabet~$R^2$.
Thus for $x=(a_1,b_1\mid  \ldots\mid a_n,b_n)$ we have
$\wtH(x) =|\{i\mid (a_i,b_i)\neq(0,0)\}|=\wts(\gamma^{-1}(x))$.
Moreover, on $(R^2)^n$ we define the inner product $\inner{x}{y}:=\inners{\gamma^{-1}(x)}{\gamma^{-1}(y)}$ for all $x,y\in(R^2)^n$.
Explicitly, this amounts to
\begin{equation}\label{e-J}
  \inner{(x_1,\ldots,x_n)}{(y_1,\ldots,y_n)}=\sum_{i=1}^n x_iJy_i\T,
    \ \text{ where }J=\begin{pmatrix}0&-1\\1&0\end{pmatrix}
\end{equation}
and where $x_i,y_i\in R^2$.
All of this shows
\begin{equation}\label{e-ffhatw}
   f \text{ is $\wts$-preserving }\Longleftrightarrow \hat{f}\text{ is $\wtH$-preserving}
\end{equation}
and
\begin{equation}\label{e-ffhatp}
   f \text{ preserves }\inners{\cdot}{\cdot}\Longleftrightarrow \hat{f}\text{ preserves }\inner{\cdot}{\cdot}.
\end{equation}

This allows us to translate the question of symplectic isometries to linear maps preserving the Hamming weight and the
above inner product.

We first describe the symplectic isometries on the entire ambient space~$R^{2n}$.
The following naturally generalizes the fact that every Hamming isometry $f:R^n\longrightarrow R^n$ is a monomial map, that is, given by a matrix representation of the form $DP$, where~$P$ is an $n\times n$-permutation matrix and $D\in\GL_n(R)$ a diagonal matrix with units on the diagonal.

\begin{theo}\label{T-IsoR2n}
Let $f:R^{2n}\longrightarrow R^{2n}$ be a linear map.
Then~$f$ is a symplectic isometry iff the matrix representation of~$\hat{f}$ in $R^{2n\times 2n}$
is a block matrix consisting of $2\times2$-blocks such that in each block row and block column there is exactly one nonzero block and that block is in
$\SL_2(R):=\{M\in R^{2\times 2}\mid \det(M)=1\}$. In other words, the matrix representation of~$\hat{f}$ is of the form
\[
  \diag(A_1,\ldots,A_n)(P\otimes I_2),\ \text{ where }A_i\in\SL_2(R)\text{ and }P\in\cS_n,
\]
where $\cS_n$ is the group of $n\times n$-permutation matrices and $\otimes$ denotes the Kronecker product.
\end{theo}

\begin{proof}
By~\eqref{e-ffhatw} and~\eqref{e-ffhatp} $f$ is a symplectic isometry iff~$\hat{f}$ is $\wtH$-preserving and preserves
$\inner{\cdot}{\cdot}$. Hence we may simply consider~$\hat{f}$.
Before proving the equivalence, we note that for any matrix $A\in  R^{2\times 2}$ we have
\[
  AJA\T=\begin{pmatrix}0&-\det(A)\\\det(A)&0\end{pmatrix}.
\]
Hence
\begin{equation}\label{e-AJA}
  AJA\T=J\Longleftrightarrow A\in\SL_2(R).
\end{equation}
As for the stated equivalence note first that the $\wtH$-preserving maps on~$(R^2)^n$ are exactly the monomial maps, that is, a permutation of the 
$R^2$-entries along with automorphisms applied to each entry.
This follows immediately by using the standard basis vectors.
In other words,~$\hat{f}$ is $\wtH$-preserving iff there exist matrices $A_i\in\GL_2(R)$ and a
permutation matrix~$P_{\sigma}\in\cS_n$ such that
\[
  \hat{f}(x)=x\, \diag(A_1,\ldots,A_n)(P_{\sigma}\otimes I_2)\text{ for all }x\in(R^2)^n.
\]
Here~$P_{\sigma}$ is the permutation matrix  satisfying $(z_1,\ldots,z_n)P_\sigma=(z_{\sigma(1)},\ldots,z_{\sigma(n)})$.
For this map~$\hat{f}$ we compute with the aid of~\eqref{e-J}
\begin{equation}\label{e-fhatinner}
   \inner{\hat{f}(x)}{\hat{f}(y)}=\sum_{i=1}^n x_{\sigma(i)}A_{\sigma(i)}JA_{\sigma(i)}\T y_{\sigma(i)}\T=\sum_{j=1}^n x_jA_jJA_j\T y_j\T,
\end{equation}
where $x=(x_1,\ldots,x_n)$ with $x_i\in R^2$ and similarly for~$y$.
Again with the help of the standard basis vectors we obtain
\[
  \hat{f}\text{ preseves }\inner{\cdot}{\cdot}\Longleftrightarrow A_jJA_j\T=J\text{ for all }j\Longleftrightarrow A_j\in\SL_2(R)\text{ for all }j=1,\ldots,n,
\]
where the second step follows from~\eqref{e-AJA}.
This concludes the proof.
\end{proof}

\begin{exa}\label{E-Tausigmai}
Consider the symplectic isometries introduced in Definition~\ref{D-SEquiv}.
One easily verifies that the isometry~$\tau_\sigma$ transforms to $\widehat{\tau_{\sigma}}$ with matrix representation $P_{\sigma}\otimes I_2$.
The isometry~$\tau_i$ transforms to $\widehat{\tau_i}$ with matrix representation $\diag(I,\ldots,I,J,I,\ldots,I)\in\GL_{2n}(R)$, with $J$ at the $i$-th diagonal position.
\end{exa}

We now turn to the question whether a symplectic isometry $f:C\longrightarrow R^{2n}$ extends to a symplectic isometry
on all of~$R^{2n}$.
One may recall from \cite[Theorem 6.3]{Wo99} that when~$R$ is a Frobenius ring and $M\subseteq R^n$ a submodule,
every Hamming isometry
$M\longrightarrow R^n$ extends to a Hamming isometry $R^n \longrightarrow R^n$, and therefore is a monomial map.
This is usually referred to as MacWilliams Extension Theorem.
However, the extendibility of Hamming isometries is not always possible if the alphabet is not a Frobenius ring or Frobenius bimodule.
For instance, if we replace the alphabet~$R$ by~$R^2$, a Hamming isometry is not a monomial map in general.
This is due to the fact that $\soc(R^2)$ is not a cyclic module; see \cite[Theorem 5.2]{Wo09}.

It turns out that in our situation the extension property fails, that is, a symplectic isometry $f:C\longrightarrow R^{2n}$ may not extend to such a map on the entire ambient space~$R^{2n}$. In other words, $f$ is not of the form as in Theorem~\ref{T-IsoR2n}.

\begin{exa}\label{E-NonExt}
Consider the binary code $C\subseteq\F_2^8$, which is generated by either of the matrices
\[
   G_1=\left(\!\!\begin{array}{cccc|cccc}
               1&0&1&1&0&1&0&0\\0&1&0&1&1&0&0&0\\0&0&0&0&1&0&1&0\\0&0&0&0&1&1&0&1\end{array}\!\!\right),\quad
   G_2=\left(\!\!\begin{array}{cccc|cccc}
               1&1&1&0&1&0&1&1\\0&0&0&0&1&1&0&1\\0&1&0&1&0&1&0&1\\0&0&0&0&0&1&1&1\end{array}\!\!\right).
\]
One easily checks that $C$ is self-dual under the symplectic inner product.
Consider the linear map $f:C\longrightarrow C$ that maps the $i$-th row of~$G_1$ to the $i$-th row of~$G_2$.
One straightforwardly checks that~$f$ preserves the symplectic weight.
It clearly also preserves the symplectic inner product as the code is self-dual.
In order to see that~$f$ does not extend to a symplectic isometry n~$\F_2^8$, we transform the map to
$\hat{f}=\gamma\circ f\circ\gamma^{-1}$; see~\eqref{e-Diagram}.
In other words, we consider the code $\gamma(C)$.
It is generated by either of the matrices $H_i=\gamma(G_i)$, where we apply~$\gamma$ row-wise to the matrices.
Hence
\[
H_1 = \left(\!\!
	\begin{array}{cc|cc|cc|cc}
    1 & 0 & 0 & 1 & 1 & 0 & 1 & 0 \\
    0 & 1 & 1 & 0 & 0 & 0 & 1 & 0 \\
    0 & 1 & 0 & 0 & 0 & 1 & 0 & 0 \\
    0 & 1 & 0 & 1 & 0 & 0 & 0 & 1
	\end{array}\!\!\right), \quad
H_2 = \left(\!\!\begin{array}{cc|cc|cc|cc}
    1 & 1 & 1 & 0 & 1 & 1 & 0 & 1 \\
    0 & 1 & 0 & 1 & 0 & 0 & 0 & 1 \\
    0 & 0 & 1 & 1 & 0 & 0 & 1 & 1 \\
    0 & 0 & 0 & 1 & 0 & 1 & 0 & 1
    \end{array}\!\!\right).
\]
Now $\hat{f}:\gamma(C)\longrightarrow \gamma(C)$ is the linear map that sends the $i$-th row of~$H_1$ to the $i$-th row of~$H_2$.
Suppose~$f$ extends to a symplectic isometry on~$\F_2^8$.
Then Theorem~\ref{T-IsoR2n} tells us that there exist matrices $A_1,\ldots,A_4\in\SL_2(\F_2)$ and a permutation
matrix $P_{\sigma}\in\cS_4$ such that
\[
    H_1\diag(A_1,A_2,A_3,A_4)(P_{\sigma}\otimes I_2)=H_2.
\]
Writing $H_1=(X_1\mid X_2\mid X_3\mid X_4)$, where $X_i$ denotes the $i$-th $4\times 2$-block of~$H_1$, and similarly
$H_2=(Y_1\mid Y_2\mid Y_3\mid Y_4)$, the above reads as
\[
  (X_{\sigma(1)}A_{\sigma(1)}\mid X_{\sigma(2)}A_{\sigma(2)}\mid X_{\sigma(3)}A_{\sigma(3)}\mid
      X_{\sigma(4)}A_{\sigma(4)})=(Y_1\mid Y_2\mid Y_3\mid Y_4).
\]
But this is impossible because of, for instance, the $2\times 2$-zero block in~$Y_1$.
Note that since $\SL_2(\F_2)=\GL_2(\F_2)$ this argument shows that~$f$ does not even extend to a
$\wts$-preserving map.
\end{exa}

Whenever the extension property fails one may ask about further information on the isometry group of the code under
consideration.
Let us make this more precise.

For obvious reasons, we call symplectic isometries of~$R^{2n}$  described in Theorem~\ref{T-IsoR2n} $\SL_2(R)$-\emph{monomial maps}.
For a submodule $C\subseteq R^{2n}$ define the groups
\begin{align*}
   \MonSL(C)&:=\{f:C\longrightarrow C\mid f\text{ is an $\SL_2(R)$-monomial map} \},\\
   \Symp(C)&:=\{f:C\longrightarrow C\mid f\text{ is a symplectic isometry}\}.
\end{align*}
Thus $\MonSL(R^{2n})=\Symp(R^{2n})$, but in general $\MonSL(C)\subsetneq\Symp(C)$.
This gives rise to the following open problem.

\begin{question}\label{Q1}
For a self-orthogonal submodule $C\subseteq R^{2n}$, how different can $\MonSL(C)$ and $\Symp(C)$ be?
\end{question}

In \cite{Wo16} the author considers an analogue question for additive codes over finite fields, and shows that the difference can be as big as possible. Namely, there exist codes with trivial group of monomial maps and maximal group of isometries.
In our situation, it is difficult to bring into play the self-orthogonality and therefore Question \ref{Q1} is wide open. The case for binary stabilizer codes is under consideration in \cite{Pl17}.

We also wish to remark that one can find examples of symplectic isometries $f:C_1\rightarrow C_2$ between self-orthogonal codes $C_1,C_2$ that do not extend to a symplectic isometry
$\tilde{f}:\tilde{C}_1\rightarrow\tilde{C}_2$, for \emph{any}  self-dual codes~$\tilde{C}_i$ containing~$C_i$.
This leads us to the following open problems, with which we close this paper.

\begin{question}\label{Q2}
Let $C\subseteq R^{2n}$ be a self-orthogonal submodule. How different can the group $\Symp(C)$  be from the group
$\Symp(C^\perp,C):=\{f\in\Symp(C^\perp)\mid f(C)=C\}$?
\end{question}

\begin{question}\label{Q3}
Let $C\subseteq R^{2n}$ be a self-orthogonal submodule. How different can the group $\MonSL(C^{\perp},C) := \{f \in \MonSL(C^{\perp}) \mid f(C) = C\}$  be from the group
$\Symp(C^\perp,C):=\{f\in\Symp(C^\perp)\mid f(C)=C\}$?
\end{question}

Note that for self-dual codes, Questions~\ref{Q3} and~\ref{Q1} are identical.

\bibliographystyle{abbrv}
\bibliography{literatureAK,literatureLZ}

\end{document}